\documentclass[12pt]{article}
\usepackage{color}
\usepackage{graphicx}
\usepackage{amsfonts}
\usepackage{latexsym,amssymb,amsmath,pslatex}
\usepackage{url}
\usepackage{authblk}
\newtheorem{definition}{Definition}

\newtheorem{lemma}{Lemma}
\newtheorem{theorem}{Theorem}
\newtheorem{corollary}{Corollary}
\newtheorem{claim}{Claim}

\newenvironment{proof}{\noindent{\bf Proof.}}{\hfill \qed \vskip 5pt}
\def\qed{\hfill\rule{2mm}{2mm}}

\newcommand{\la}{\leftarrow}
\begin{document}

\title{Optimal online and offline algorithms for robot-assisted restoration of barrier coverage
\thanks{This work was partially supported by NSERC grants}}
\author[1]{J. Czyzowicz}
\author[2]{E. Kranakis} 
\author[3]{D. Krizanc}
\author[4]{L. Narayanan}
\author[4]{J. Opatrny}
\affil[1]{{\small D\'{e}p. d'informatique, Universit\'{e} du Qu\'{e}bec en 
Outaouais, Gatineau, QC, Canada, {\tt Jurek.Czyzowicz@uqo.ca}}} 
\affil[2]{{\small School of Computer Science, Carleton University, Ottawa, ON, 
Canada {\tt kranakis@scs.carleton.ca}}}
\affil[3]{{\small Dept. of Math. and Computer Science,
Wesleyan University, Middletown, CT, USA {\tt dkrizanc@wesleyan.edu}}}
\affil[4]{{\small Dept. of Comp. Science and Software Eng.,
Concordia University, Montreal, QC, Canada {\tt lata,$\;\;$opatrny@cs.concordia.ca}}}

\date{}                      

\maketitle

\begin{abstract}
Cooperation between mobile robots and wireless sensor networks is a line 
of research that is currently attracting a lot of attention.  
In this context, we study the following problem of barrier coverage by 
stationary wireless sensors that are assisted by a mobile robot with the 
capacity to move sensors. Assume that 
$n$ sensors are initially arbitrarily distributed on a line segment barrier. 
Each sensor is said to cover the portion of the barrier that intersects with 
its sensing area. Owing to incorrect initial position, or the death of some of 
the sensors, the barrier is not completely covered  by the sensors. We  
employ a mobile robot to move the sensors to final positions on the barrier 
such that barrier coverage is guaranteed. We seek algorithms that minimize 
the length of the robot's trajectory, since this allows the restoration of
barrier coverage as soon as possible. 
We give an optimal linear-time offline algorithm that gives a 
minimum-length trajectory for a robot that starts at one end of the barrier and achieves the restoration of barrier coverage. 
We also study two different online models: one in which the online robot does 
not know  the length of the barrier in advance, and the other in which the 
online robot knows the length of the barrier. 
For the case when the online robot does not know the length of the barrier, 
we prove a tight bound of $3/2$ on the competitive ratio, and 
we give a tight lower bound of $5/4$ on the competitive ratio in the 
other case. Thus for each case we give an optimal online algorithm. 
\end{abstract}


\section{Introduction}

Mobile robots and wireless sensor networks are related areas of research  that have largely been studied by different communities of researchers. Recently, there has been increasing interest in the possibilities uncovered by utilizing {\em both}  technologies \cite{Kouba14}: what if mobile robots and wireless sensors could {\em cooperate} to solve problems and perform tasks? 
Environments where autonomous networked entities such as robots and sensors 
cooperate to achieve a common goal are sometimes called 
{\em mixed-mode environments} and have been the subject of 
several recent research events, e.g.,  \cite{robosense13,GkmM12}. 

In this paper, we study a related mixed-mode problem for barrier coverage. Assume $n$ stationary sensors have initial positions on a line segment barrier. 
Owing to incorrect initial placement, or the death of some sensors due to 
battery failure or a disaster, the barrier is not completely covered by the 
sensors. An illustration is given in Figure~\ref{fig:ex}(a), where the segment of the 
barrier covered by a sensor is represented as a box. 
\begin{figure}[!htb]
\begin{center}
\includegraphics[width=11cm]{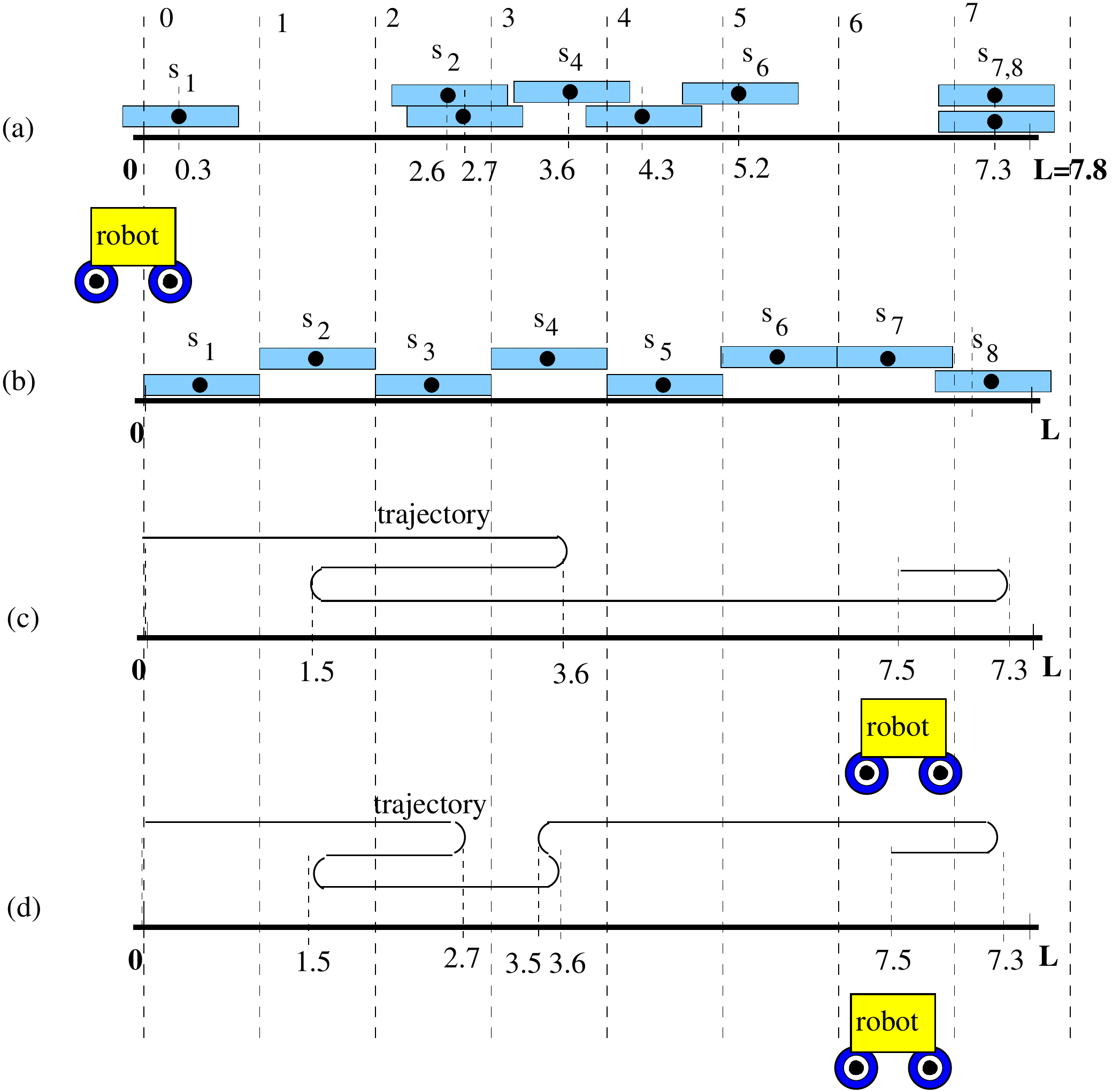}
\end{center}
\caption{Robot-assisted restoration of barrier coverage 
problem with sensor range equal to 0.5: 
(a) the initial configuration on segment $[0,L]$ with gaps in coverage, (b) a
possible solution, (c) and (d) give examples of trajectories 
that could be followed by the robot to obtain the final configuration in (b).}
\label{fig:ex}
\end{figure}

The task of the mobile robot is to walk along the barrier segment, 
pick up and move sensors to final positions such that barrier coverage 
is restored, {\em i.e.}, in their final positions, the sensors collectively cover the entire line segment barrier as in Figure \ref{fig:ex}(b). Note that the final positions that achieve barrier coverage are not unique. 
Since sensors may need to be moved in different directions, i.e. left or right,  to assure coverage, the robot may sometimes need to turn or change direction 
in order to restore coverage. 
The robot may decide to resolve the gap as soon as possible, as late as 
possible, or some time in between. 
The robot thus follows a certain {\em trajectory}, which can be specified by 
the starting point, and a sequence of points where the robot alternately turns 
left and right before it reaches its termination point. Given the initial configuration 
of Figure~\ref{fig:ex}(a), two of the possible trajectories that achieve the same final 
positions of sensors are shown in Figures \ref{fig:ex}(c) and (d).  
The time needed to restore  barrier coverage  is clearly related to the length of the 
robot's trajectory, which in turn depends on the knowledge it has of the initial 
positions of sensors. 
The problem we are interested in is finding an {\em optimal} trajectory for 
the mobile robot in order to achieve barrier coverage as fast as possible. 

\noindent
{\bf Sensor relocation model.}
In the sequel we define the capabilities of the sensors and the robot, as well as the trajectory of the robot. 

\noindent
{\em Sensors. }
Assume that $n$ {\em sensors} $s_1,s_2,\ldots,s_n$ 
are distributed on the line segment $[0, L]$ of length $L$ with
endpoints $0$ and $L$ in locations $x_1\leq x_2\leq \ldots\leq x_n$. 
The range of all sensors is assumed to be identical, and is equal to 
a positive real number $r>0$. 
Thus sensor $s_i$ in position $x_i$ 
defines a closed interval $[x_i-r, x_i+r]$ of length $2r$ 
centered at the current position $x_i$ of the sensor, in which it can detect 
an intruding object or an event of interest. See Figure~\ref{fig:ex}(a) for an illustration of a problem instance. 
We say that the sensor
{\em covers} the closed interval $[x_i-r, x_i+r]$.
We assume that the total range of the sensors is sufficient to cover the
entire line segment $[0, L]$, i.e., $2rn \geq L$. 
We define a {\em gap} to be a closed subinterval $G$ of $[0,L]$ such that
no point in $G$ is within the range of a sensor. Clearly,
an initial placement of the sensors  may have gaps.
The sensors provide {\em complete coverage} of  $[0,L]$ if they leave no gaps. 

\noindent
{\em Robot.} 
There is a {\em mobile robot} that can move the sensors to positions that
guarantee coverage of the entire line segment. 
We assume that the robot can {\em pick, carry, move} and {\em drop/deposit} sensors
from any initial position to any desired position on the line segment.
There is no constraint on the direction and number of
turns it can take (left or right) so as to pick and/or drop sensors, and no 
restriction on where in the line segment it can drop the sensors. 
We study the case when sensors are small enough and thus the robot 
can potentially carry all the sensors it needs at the same time. 

\noindent
{\em Robot trajectory and length.}
 Our goal is to provide offline and online algorithms so as to minimize the time taken 
to restore barrier coverage. Assuming constant speed, we measure this  by the 
distance 
travelled by the robot from its starting position to complete  
the task of moving sensors to positions which guarantee complete coverage of the barrier.  
We assume that the mobile robot starts at position $0$ and moves to the right. 
At some point it can turn and move left, then again turn and move right and 
so on. Thus its trajectory can be specified as a sequence of points on the 
barrier: $[t_0=0, t_1,t_2, \ldots, t_m]$, where the points $t_1,t_3,\ldots$ 
are the points where the robot turns left,  
the points $t_2,t_4,\ldots$  are the points where the robot turns right, 
and finally,  the point $t_m$ is the {\em termination point} of the 
trajectory. Therefore, $t_i>t_{i-1}$ for 
all odd $i$ while $t_i < t_{i-1}$ for all even $i$ where $0 < i \leq m$, and
the robot's trajectory is  the sequence of line segments 
$[0, t_1], [t_1, t_2], \cdots, [t_{m-1}, t_m]$.
The length of the trajectory is defined as $\Sigma_{i=1}^m |t_i - t_{i-1}|$.

We seek  algorithms that calculate  an {\em optimal} trajectory
for the mobile robot that ensure barrier coverage, i.e. a trajectory of 
smallest possible length. A mobile robot using an {\em offline} algorithm to
calculate its trajectory is 
assumed to know all the initial positions of sensors before starting its 
trajectory. On the other hand, a robot using an  {\em online} algorithm  
knows about sensors only in the parts of the barrier segment where it has 
already travelled. Specifically, an online robot discovers the presence or absence of a 
sensor at position $x$  only when reaching $x$. Therefore, at the start of 
the algorithm, such a robot has no knowledge about any of the 
sensors' positions. It can of course
remember any sensors that it has seen previously. 

\noindent
{\bf Related work.}
Barrier coverage using wireless sensors has been the subject of intensive 
research in the last decade 
\cite{barriercoverageNodeDegree,barriercoverage05,barrierCoveragePlane}. 
Some papers assumed randomized deployment of sensors on the barrier and 
analyzed the probability of barrier coverage. Other papers have studied the 
case of relocatable sensors  \cite{SC10,Teng07}, which 
start at arbitrary positions and can move to final positions that achieve 
barrier coverage. Centralized algorithms for minimizing the maximum and 
average movement of sensors were studied in \cite{swat2012,adhocnow2009} and \cite{adhocnow2010} respectively. 
Multiple barriers were studied in 
\cite{tcs2009},
and distributed algorithms for barrier coverage were given for the first time in \cite{HesariKKPNOS13}.

Charikar~et~al.~\cite{charikar2001algorithms}
consider the $k$-delivery TSP problem for transporting efficiently
$n$ identical objects, placed at arbitrary initial locations, 
to $n$ target locations with a vehicle that can carry at 
most $k$ objects at a time. 
Chalopin~et~al.~\cite{wid2013} provide hardness results, exact, approximation, and 
resource-augmented algorithms for the problem of whether there is a schedule of 
agents' movements 
that collaboratively deliver data from 
specified sources of a network to a central repository. Our problem differs both in being uncapacitated, 
and in the fact that the locations 
of the sources and targets are not known in advance.

Online vehicle routing problems and the online travelling salesman problem  have been studied previously; see  \cite{jaillet-wagner} for a survey. 
Our problem and our conception of online are quite different: the locations the 
robot needs to deposit sensors are not pre-determined, and we assume an 
online robot discovers the positions of sensors as it moves along the barrier. 

Cooperation between mobile robots and wireless sensors is a relatively new 
research area and has been explored in several research events in the last 
couple of years \cite{GkmM12,Kouba14,robosense12,robosense13}. The authors 
of  \cite{DSKZ06} and \cite{SSL07} use information obtained from wireless sensors 
for the problem of localization of a mobile robot.  In \cite{JS01}, mobile 
robots and stationary sensors cooperate in a target tracking problem: 
stationary sensors track moving sensors in their sensor range, while mobile 
robots explore regions not covered by the fixed sensors. A common evaluation 
platform for mixed-mode environments incorporating both mobile robots and 
wireless sensor networks is described in \cite{Kropff08}.

\noindent{\bf Our results.}
We give a linear time offline algorithm that computes an optimal trajectory 
for a robot starting at an endpoint of the barrier to restore barrier coverage. 
For the online case, we show that when the robot does not know the length of 
the barrier and recognizes the end of it only when reaching it, 
any algorithm must have a competitive ratio of at least $3/2$. 
We give a simple algorithm that matches this bound. When the robot does know 
the length of the barrier, we show a lower bound  of $5/4$ on the 
competitive ratio of any online algorithm for the problem. We then give an 
adaptive online algorithm whose competitive ratio matches this lower bound. 

\section{Optimal Offline Algorithm}

In our offline algorithm we assume that the robot has global 
knowledge of the positions of sensors on the line segment, and  
that during the course of its
movement, can pick up and carry as many sensors as necessary
and deposit them as required. All sensors have  
identical range denoted by $r$, and the robot starts at the endpoint $0$ of
the interval $[0,L]$, and the number of given sensors is sufficient to cover 
the given interval.  

Obviously, when the barrier does not contain any gap, the trajectory 
is empty and we consider below instances containing gaps.
We begin by establishing the properties of 
optimal non-empty trajectories of the robot, 
which are crucial to the development of the algorithm. 
We say that a solution is  {\em order-preserving} if the final order
of the position of the sensors is the same as their initial positions. 
Secondly, a solution is called {\em fully stretched} if the robot 
places all sensors in {\em attached positions},  
i.e., two consecutive sensors encountered 
by the robot are placed at distance $2r$ 
and the first sensor is at distance $r$ from $0$, except possibly the sensors 
at or after the termination point $t_m$ as in the example in 
Figure~\ref{fig:ex}(b). 

\begin{lemma}{\em (Order-preserving fully stretched solution)}
\label{lm:op}
There exists an optimal trajectory for the robot that produces an 
order-preserving fully-stretched solution.
\end{lemma}
\begin{proof}
Let $s_i$ and $s_j$ be two sensors such that $x_i<x_j$ and assume that 
the robot, when following an optimal trajectory places $s_i$ to the right of
$s_j$. Clearly, on any optimal trajectory the robot encounters  $s_i$ before 
$s_j$.
If on the trajectory of the robot, the placement of  $s_i$ occurs after 
 the robot encounters $s_j$, then the robot can pick $s_j$ when 
traversing it and reverse the placements of  $s_i$ and $s_j$.
 If on the trajectory of the robot, the placement of  $s_i$ occurs before 
the robot encounters $s_j$ then the robot makes a left turn after 
encountering $s_j$ and the trajectory of the robot must cross the final 
position of $s_i$. When crossing $s_i$, the robot can replace $s_i$ by $s_j$,
and place later  $s_i$ in place of $s_j$ using the same trajectory. 
Since the sensors have identical ranges, this exchange of positions of 
$s_i$ and $s_j$ gives the same barrier coverage.

When following an order  preserving trajectory, the robot can 
pick any sensor it encounters and delay the drop of a sensor, or drop it
earlier so that the distance to the beginning of the segment is $r$ or 
to the preceding sensor is $2r$. This can be done as long as the span of the 
trajectory extends $2r$ past the previous sensor. 
\end{proof}
\begin{lemma}{\em (Three Visits Lemma)}
\label{lm2}
The trajectory of an optimal algorithm does not contain  
the same point of the line segment more than three times. 
Furthermore, the last point of the trajectory can occur in the trajectory 
at most twice.  
\end{lemma}
\begin{proof}
The idea is to demonstrate that a robot visiting a point in the line 
segment more than three times produces a non-optimal trajectory. 
This is based on the {\em early pick-up, late drop} principle,
i.e. we show that if a point $p$ is visited more than three times, we can
replace a part of a trajectory with a new, shorter  
trajectory in which pick-ups of sensors are done as early as
possible,  and drops of sensors  
are delayed to after the picking in the affected part is finished.

Assume on the contrary that a point, say $p$, on the line segment
is visited by the robot at least $k$ times, $k \geq 4$, during its optimal
trajectory $T$, see Figure \ref{fig:fg1}. 
Consider  trajectory $T_1$, a part of  trajectory $T$, 
that the robot follows from the first to the last 
visit of  $p$. Let $a$ be the rightmost point and 
$b$ be the leftmost  point of the segment $[0,L]$ occurring on $T_1$. Let $T_2$ 
be a part of trajectory $T$ from $b$  to $p$ prior to 
reaching $p$ for the first time.  Part 
$T_2$ must exist since the trajectory starts 
at $0$. Clearly, all sensors that 
the robot deposits between $b$ and $a$ when following $T$ up to the 
last visit of $p$ are from among the sensors that the 
robot carried when entering $T_2$ and those  sensors the robot 
found on the segment between $b$ and $a$.  Thus the robot  can achieve 
the same result by replacing $T_2$ and $T_1$ by  trajectory $T_3$
consisting of three  parts:
Part one of $T_3$ is segment $[b,a]$, and the robot when moving 
from $b$ to $a$ picks all sensors found there.
(This ``early pick'' allows the robot to have all sensors needed later on),
Part two is segment $[a,b]$ and the robot when moving
from $a$ to $b$ drops the sensors in the same order as 
achieved by $T_1$ and $T_2$ (the late drop part).
The third part is segment $[b,p]$ moving  from $b$ to $p$, where the robot does 
no picks or drops  of sensors, see Figure \ref{fig:fg2} (a). 

When following $T_3$, the robot  picked all sensors found between $b$ and $a$
and dropped the same sensors as when following $T_2$ and $T_1$. Thus,
when it reaches $p$ for the third time, it carries at least  the same or
more sensors than when reaching $p$ along  $T$ for the last time. 
Thus the robot then can continue along  the rest of $T$ and achieve 
the same result, while visiting $p$ only three times. 
This trajectory is shorter, a contradiction.   

If $T$ ends at $p$, then the third part $[b,p]$ of $T_3$ is not needed,
since the robot is doing nothing, and the trajectory should end at $b$. 
This implies that the last point of an 
optimal trajectory can be visited at most twice. The case when trajectory 
$T_3$ ends at $b$ is shown in  Figure~\ref{fig:fg2} (b).
\end{proof}
\begin{figure}[!htb]
\begin{center}
\includegraphics[width=6cm]{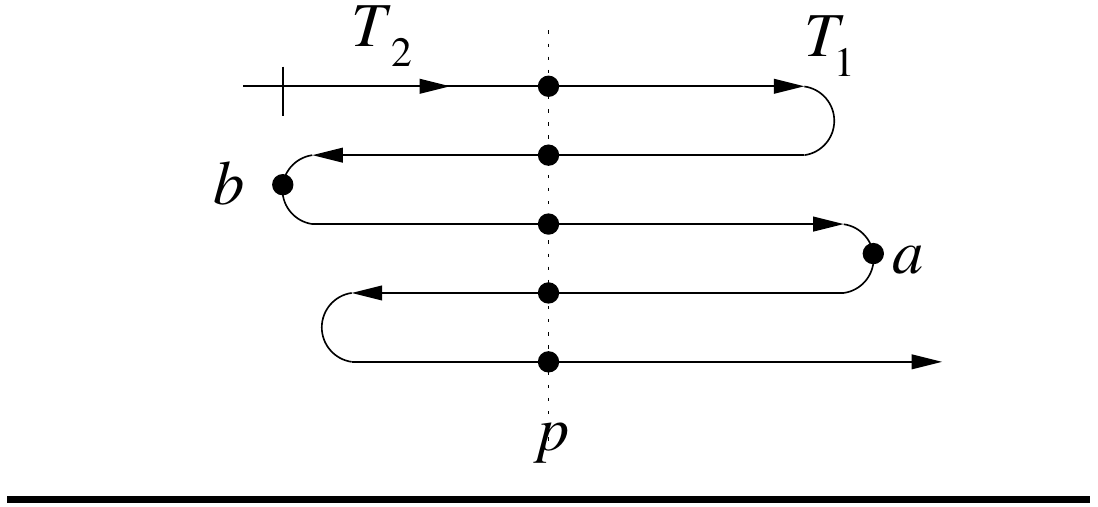}
\end{center}
\caption{More than three visits of the point $p$ by the trajectory 
of the robot.}
\label{fig:fg1}
\end{figure}
\begin{figure}[!htb]
\begin{center}
\includegraphics[width=12cm]{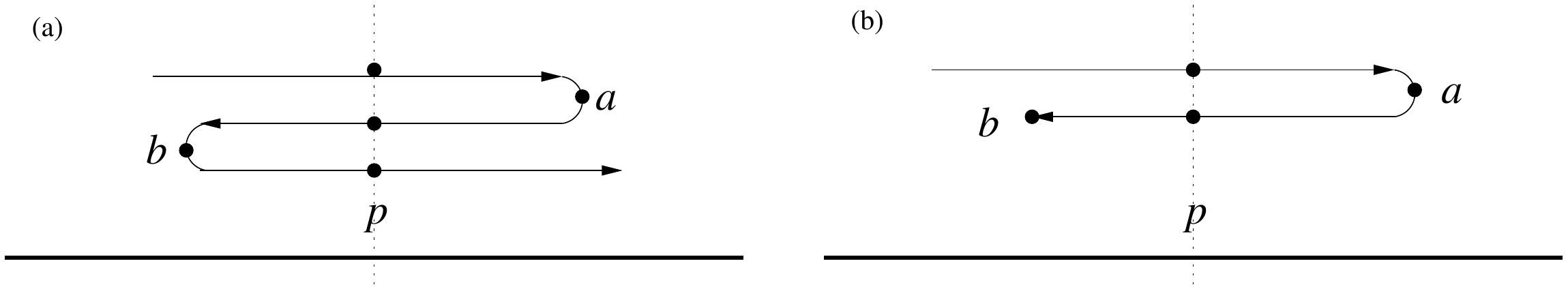}
\end{center}
\caption{Replacing the old trajectory of the robot with either of the
two depicted trajectories.}
\label{fig:fg2}
\end{figure}
Observe that the above lemmas applies to both offline and online algorithms. 
Furthermore, once a trajectory is specified, the robot can produce an 
order-preserving fully stretched solution as discussed below,
so it suffices to specify the trajectory of the robot.

Given an optimal trajectory 
$[t_0, t_1, t_2,t_3 \ldots, t_{m-1}, t_m]$, the robot makes a left turn 
at $t_1, t_3, \ldots$ and right turns at $t_2, t_4,\ldots$. 
Therefore, the segments 
$[t_2,t_1], [t_4, t_3],$ $ \ldots,$ $[t_{2i},t_{2i-1}],\ldots$ of $[0,L]$ are
traversed by the robot three times, $1\leq i \leq (m-1)/2$, 
and if $m$ is even, then the segment $[t_m, t_m-1]$, is traversed twice.
Furthermore, all these segments are pairwise disjoint, except possibly for 
the endpoints of the segments. 
We call the part of the trajectory $[t_{2i},t_{2i-1}]$,  $1\leq i \leq (m-1)/2$, 
traversed by the robot three times, 
 a {\em triple}, $t_{2i-1}$ is called its {\em left turning point} and
 $t_{2i}$ is called its {\em right turning point}.
If $m$ is even, then the segment $[t_{m}, t_{m-1}]$ is
 called a {\em double} and  $t_{m-1}$ is called its 
{\em left turning point}. Any line segment in the trajectory that is 
traversed exactly once by the robot 
is called a {\em straight  line segment} (see Figure \ref{fig:fg3}). 
When following a  straight  line segment the robot necessarily has sufficient 
supply of sensors and deposits them in attached positions. When following 
a segment of a triple or a double  for the first time, the robot picks 
all sensors found there and deposits then in attached positions when going 
back over the segment (see the proof of Lemma \ref{lm2}).   
Clearly, if two consecutive triples, or a triple and a double
share an endpoint, these two moves can be merged into a single triple, or a 
double. This observation and the preceding lemmas imply the 
following corollary. 
\begin{corollary}
\label{cor1}
There is an optimal order-preserving and fully stretched trajectory of the 
robot that produces a complete coverage of $[0,L]$ which consists of $k$ 
consecutive triples and straight line segments for some $k\geq 0$, and  
ends with a straight line segment or a double (see Figure \ref{fig:fg3}).
 \end{corollary}
\begin{figure}[!htb]
\begin{center}
\includegraphics[width=10cm]{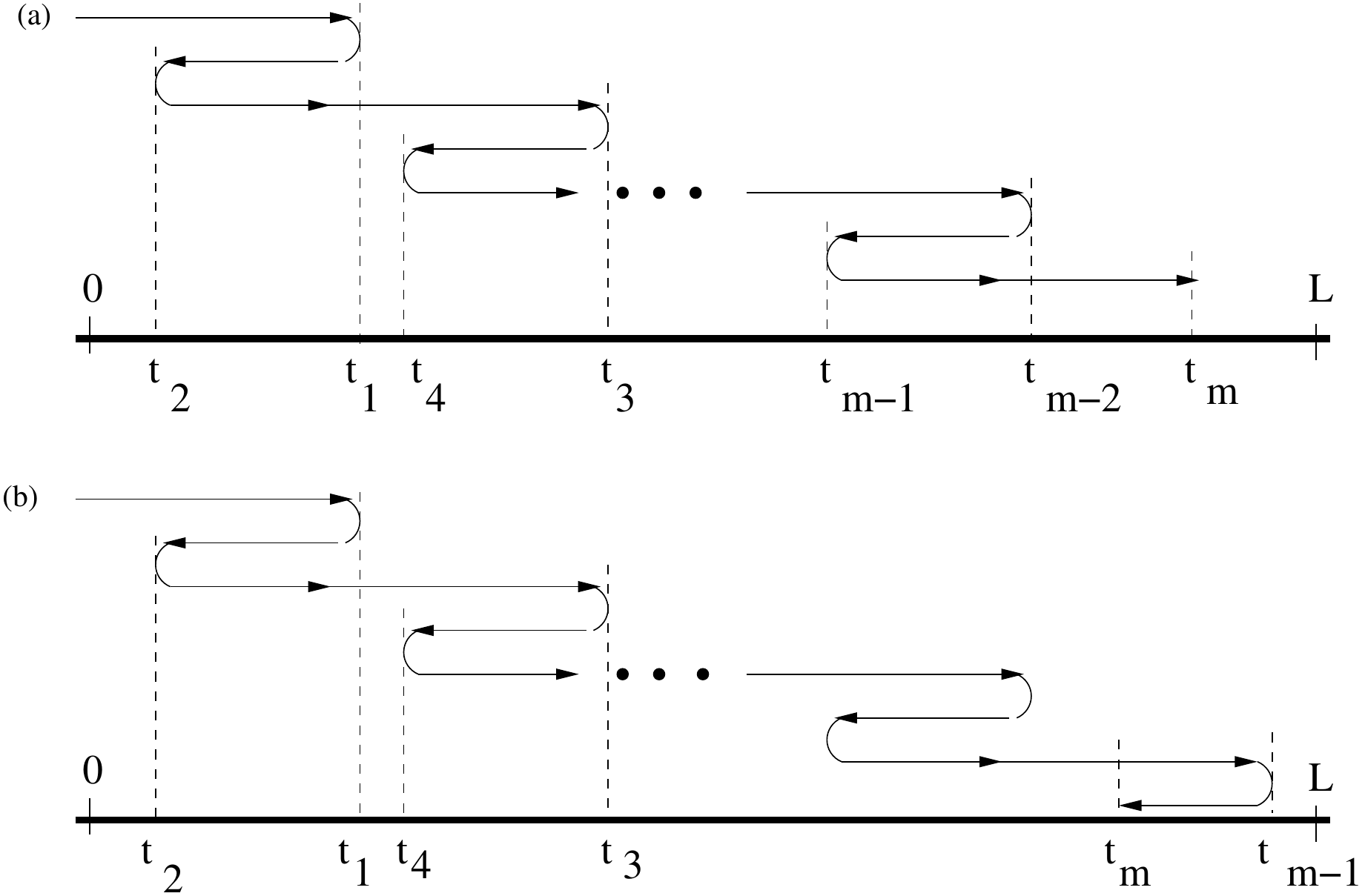}
\end{center}
\caption{Two possible shapes of an optimal trajectory.}
\label{fig:fg3}
\end{figure}
To construct an optimal trajectory of the robot we need to determine, 
from the given input instance, the ends of the triples and a double.
The following definition  of coverage balance is used to determine them.
\begin{definition}
The {\em coverage balance} of sensor $s_i$ at location  $x_i$ 
is defined to be $C_i= (2ri -r)-x_i$, 
i.e., the difference between the total length that can be covered by  
sensors $s_1,s_2,\ldots,s_i$ up to the center of $s_i$, and 
 the distance of $s_i$ to the beginning of the interval.   
\end{definition}

Consider the example in Figure~\ref{fig:ex}. The coverage balance 
of sensors listed
from left to right are $0.2, -1.1, -0.2,-0.1, 0.2, 0.3,-0.8$ 
and $0.2$. Notice that in 
the two examples of trajectories in this figure each left turn is done at a 
sensor with negative coverage balance. 
However, doing a left turn at every sensor with 
negative coverage balance would be wrong, because it could violate the 
three visits 
lemma. Similarly, doing a triple involving many consecutive sensors with 
negative coverage balance can be sub-optimal as well, 
as seen in the trajectory of Figure~\ref{fig:ex}(c).  
The following lemma specifies all potential left turning points.
\begin{lemma}
\label{lm:opt-tr}
Let $([t_0, t_1, t_2,t_3 \ldots, t_{m-1}, t_m])$ be an optimal trajectory which 
minimizes the number of triples. 
\begin{enumerate}
\item Every sensor with negative coverage balance is shifted left, and thus  
its location is in a triple  or the double  segment.
\item No triple  segment contains the location of a sensor with nonnegative 
coverage balance.
\item In a double segment
the rightmost sensor has negative coverage balance. 
\item For every triple  $[t_{2i},t_{2i-1}]$, the left-turning point 
$t_{2i-1}$ is a location of a sensor $x_j$ for some integer $j$ 
such that either $-2r<C_j<0$, or both, $-2r=C_j$ and $x_j=x_{j+1}$,  
and the coverage balance  of every other sensor located in the interval 
$[t_{2i},t_{2i-1}]$ is less than or equal to $-2r$. 
\item Let $k$ be the smallest integer such that all gaps in $[0,L]$ 
are to the left of $2rk$. Then $s_k$ is the last sensor to be moved. 
Let $c= x_k$ if $C_k<0$, else  $c= x_k+C_k$. 
If the trajectory does not end with a double  then $C_k\geq 0$, $m$ is odd,
and the termination point $t_m=c$. 
Otherwise the trajectory ends with a double, and $t_{m-1}=c$, 
i.e., the left-turning point of the double is $c$.
\end{enumerate}
\end{lemma}
\begin{proof}
Clearly, any sensor with negative coverage balance needs to be shifted to the 
left and this is done using a triple of the double . If interval  
$[t_{2i},t_{2i-1}]$ is a triple  than it cannot contain a sensor 
with nonnegative coverage balance  or otherwise we can replace one triple with 
a shorter triple  or two triples that are shorter in total.\\
If the rightmost sensor of a double  has nonnegative coverage balance, 
the double  can be made shorter by omitting this sensor from the 
double  and by making any shift of it to the right when it is encountered 
the first time.
   
Clearly, the left turning point of any triple should be a location of a 
sensor with negative coverage balance, otherwise the triple can be made 
shorter.
If $C_j \leq -2r$, and  $x_j < x_{j+1}$ then sensor $s_{j+1}$ must be shifted 
left to the position $t_{2i-1}$, or to the left of it by another triple or 
double. 
This would either allow a merge of triple  $[t_{2i},t_{2i-1}]$ with 
another triple or double, or it would contradict the three visits lemma. 
Thus,  either $-2r<C_j<0$, or $C_j= 2r$ and  $x_j= x_{j+1}$. 

Since all gaps are to the left of $2rk$, placing sensors 
$s_1,s_2,\ldots, s_k$ in attached position covers all gaps and $s_k$ is
the rightmost sensor on the trajectory. If $C_k$ is negative, the robot must 
turn right at $x_k$ for a double. If there is no double, 
$C_K\geq 0$ and by Corollary \ref{cor1}
the trajectory ends with a straight line segment,
terminating at~$c$. 
\end{proof}
Thus, by the preceding lemma, the potential left turning points in 
the example in Figure \ref{fig:ex} are the initial locations of sensors 
$s_3, s_4$, and $s_7$, but not $s_2$. 

\begin{definition} 
\label{def:A}
Let $m$ be the number of sensors whose coverage balance 
is  either $-2r<C_j<0$, or  $-2r=C_j$ and $x_j=x_{j+1}$.
The list $A$ of indices of sensors of 
{\em potential triple delimiters} is a list of pairs
$A= [(b_1, a_1), (b_2, a_2),  \cdots,  (b_m, a_m)]$ of sensor indices such that 
\begin{enumerate}
\item $a_1 < a_2 < \cdots < a_m$ are  the indices of all sensors such that
 either $-2r<C_j<0$, or  $-2r=C_j$ and $x_j=x_{j+1}$ 
\item $b_1 $ is the smallest index of a sensor with negative coverage balance, 
and   for  $1<i\leq m$ the value of $b_i$ is the smallest index 
larger than  $a_{i-1}$ with negative coverage balance.
\end{enumerate}
\end{definition}
\begin{lemma}
\label{lm:tpt}
Let $A$ be the list  of indices of sensors of potential triple 
delimiters,  $m$ 
be the number of pairs in $A$, and $c$ be defined as in Lemma \ref{lm:opt-tr}. 
There is an optimal, order preserving, fully-stretched trajectory such that for
some integer $j$, $0\leq j\leq m$,  
\begin{enumerate}
\item the trajectory contains $j$ triples  
$[x_{b_{i}} + C_{b_i},x_{a_{i}}]$,  $1\leq i\leq j$,
\item 
If $j<m$ then the trajectory ends
with a double,  $[x_{b_{j+1}} + C_{b_{j+1}},c]$, 
otherwise the trajectory ends with a straight line and its termination point 
is $c$.
\end{enumerate}
\end{lemma}
\begin{proof}
 Consider an optimal trajectory $T$ which minimizes the number of 
triples.  
Each  pair in list $A$ gives an interval of indices of sensors with negative 
coverage balance.
According to Lemma \ref{lm:opt-tr}, these sensors must be moved left by a 
triple or a double. Furthermore the left turning point of any triple 
is the location of a sensor $s_{a_i}$. Thus, if $T$ contains $j$ triples 
then $1\leq j\leq m$, and the triples are as stated in the lemma. 
If $j<m$ then all remaining sensors with negative coverage must be shifted 
left by a double, and, by Lemma \ref{lm:opt-tr}, this move is 
 $[x_{b_{j+1}} + C_{b_{j+1}},c]$. 
If $j=m$ then the trajectory must end at $c$.
\end{proof}

The main idea of our offline algorithm is to calculate the list $A$ of 
potential triple delimiters as defined in Definition \ref{def:A}.
Let $T_j$ be a trajectory that uses triples on the first $j$ 
pairs of $A$, $0\leq j \leq m$, and one double 
if $j<m$. We define the overhead $o_j$ of a trajectory $T_j$ to be the 
difference 
between the length of $T_j$ and the straight line trajectory. Clearly,

\[o_j = \left\{\begin{array}{ll}
c- x_{b_{j+1}}-C_{b_{j+1}}+\sum_{i=1}^j 2(x_{a_i} -x_{b_i} -C_{b_i}) 
\mbox{ for } 1 < j < m,\\ 
\sum_{i=1}^m 2(x_{a_i} -x_{b_i} -C_{b_i}) \mbox{ for }j=m
\end{array}\right.\]

The algorithm calculates the overhead of $T_j$ trajectories for $1\leq j\leq m$ 
and chooses the trajectory with the minimum overhead. 
By Lemma \ref{lm:tpt}, the trajectory with the minimum overhead is optimal.
Thus a robot finds the coordinates of an optimal trajectory
by executing  the offline algorithm below. 
\begin{tabbing}
xxxx\=xxx\=xxx\=xxx\=xxx\=xxx\=xxx\=xxx\=xxxxxxxxxxxxxxxxxx\=\kill
 {\bf Offline} Algorithm\\
\line(1,0){440}\\
\> {\bf Input:}
 the length $L$ of the segment, the number $n$ of sensors,\\ 
\>\>their initial locations $x_1\leq x_2 \leq \ldots \leq x_n$, 
and their range $r$;\\
\> {\bf Output:} the trajectory points for the robot.\\ 
\line(1,0){440}\\
{\bf 1} {\bf Scan} $x_1,x_2,\ldots,x_n$ for gaps;\\
{\bf 2} {\bf if} gaps exist {\bf then} \\
\>{\bf 2.1}\> {\bf Compute } the smallest integer  $k$
such that all gaps are to the left of $x_k$;\\ 
\>{\bf 2.2} \> {\bf Compute } the sequence  $C_i= (2ri -r)-x_i$, 
$1\leq i \leq k$;\\
\>{\bf 2.3} \> {\bf if} $C_k<0$ {\bf then} $c\la x_k$; {\bf else} 
$c\la x_k+C_k$;\\
\>\>\> // $c$ is the potential left-turning point of a double.\\ 
\>{\bf 2.4} \> {\bf Scan} the sequence  $C_1, C_2, \ldots,C_k$ and\\ 
\>\>\>{\bf compute}  
list $<A, B> = [(b_1,a_1), (b_2,a_2),  \cdots,  (b_m,a_m)]$;\\ 
\>\>\> // potential triple delimiters\\
\>{\bf 2.5} 
$o_j\la (c- x_{b_{j+1}}-C_{b_{j+1}})+\sum_{i=1}^j 2(x_{a_i} -x_{b_i} -C_{b_i}) $,
$1\leq j\leq m-1$; // $T_j$ overhead\\
\>{\bf 2.6} $o_m\la \sum_{i=1}^m 2(x_{a_i} -x_{b_i} -C_{b_i}) $; // $T_m$ overhead\\
\>{\bf 2.7} {\bf Compute}  $min\{ o_1,o_2, \ldots ,o_m\}$;
and its index $k$;\\
\>{\bf 2.8} {\bf Output} $x_{a_{1}},x_{b_{1}}+C_{b_{1}},x_{a_{2}}
,x_{b_{2}}+C_{b_{2}},\ldots, x_{a_{k}},x_{b_{2}}+C_{b_{2}}$;\\
\>\> //the sequence of left/right turning points of the optimal trajectory,\\
\>{\bf 2.9 If} $k<m$ {\bf then} {\bf Output} {\em there is a double  from} 
$c$ {\em to} $x_{b_{k+1}} +C_{b_{k+1}}$; \\
\>\>{\bf else the trajectory ends at $c$;}\\
$\;\;\;$
{\bf else} $[0,L]$ is initially completely covered, robot does nothing;\\ 
\end{tabbing} 
Since algorithm {\em Offline} calculates the overheads of all 
trajectories that satisfy Lemma \ref{lm:tpt} and picks 
the one with the smallest overhead, the Corollary~\ref{cor1} and 
Lemma~\ref{lm:tpt} imply that the selected trajectory is  optimal. Clearly, 
all calculations in each step are of $O(n)$ complexity. Thus we have 
the following theorem.
\begin{theorem}
\label{thm1}
Assume we are
given $n$ sensors in the line segment $[0, L]$
and a robot with starting position $0$.  Algorithm
{\em Offline} computes an optimal trajectory 
for the robot to follow in $O(n)$ time. 
\end{theorem}

 \section{Optimal Online Algorithms}
We now consider online algorithms for restoration of barrier coverage by a 
robot. 
For the online algorithm  we assume that the robot starts at position 0, 
it can move along the given line segment, but
the robot does not know the positions of sensors until it comes upon them.
As usual, we define the competitive ratio 
of an online algorithm as the length of the trajectory of 
the online algorithm divided by the length of the 
trajectory of the optimal offline algorithm. 

At the outset, observe that on the input instance where the sensors are placed 
in such a way that there is no gap in the barrier coverage, 
the offline algorithm produces a trajectory of length 0, 
while the online algorithm must traverse 
the entire barrier segment to ensure that the barrier is covered. Thus no 
online algorithm can have a bounded competitive ratio. To provide a more 
meaningful comparison of online with offline algorithms, 
we only consider below 
input instances where there is a gap in coverage at the very end of 
the barrier, that is, the point $L$ is uncovered. On such instances, all 
valid robot trajectories must have length at least $L-r$.
We also consider the possibility that $L$, the length of the barrier, is not
known to the robot and the robot will find it out only when reaching the end
of the barrier.
Since the performance of online algorithms depends on the knowledge 
of $L$, we consider the two possibilities 
separately. We use below the notion of potential 
left and right turning points as defined in the previous section. 

When the value of $L$ is unknown to the robot we show the following result.
\begin{theorem}
\label{thm3}
Assume that the robot does not know the length $L$ of the barrier 
$[0, L]$.
For any $0 < \epsilon \ll 1$, the competitive ratio of any online algorithm is at least $\frac{3}{2} - \epsilon$.
Furthermore, there is an online algorithm with competitive ratio at most $\frac{3}{2} $.
\end{theorem}
\begin{proof}
Assume there is an online algorithm $\cal{A}$ with competitive ratio 
$3/2 - \epsilon$ for some $\epsilon>0$. We give an adversary argument. 
Start with an input that has no sensors until position $x=2ir$ where there 
are $i>0$ sensors for some $i$ to be specified later. Clearly 
there are just enough sensors at $x$ to cover the segment $[0,x]$. 
Following this, the adversary starts placing sensors in attached position
starting at position $x+2r$. 
The robot has to make a turn at some point 
$y \geq x$ to cover the gap in coverage before $x$. If it does not make a turn before $6x$, 
the adversary can set $L=6x$, and 
the robot must do a double to the beginning, see 
Figure~\ref{fig:fg4ab} (a).
The trajectory produced by $\cal{A}$ has length at least 
$2L-r=12x-r$, while the offline algorithm covers the gap before 
$x$ with a triple  from $x$ using a trajectory of at most 
$3x-2r+5x= 8x-2r$.
This gives the competitive ratio of at least $(12x-r)/ (8x-2r) > 3/2$.

If the robot turns at any point $y$ such that  $x \leq y < 6x$, 
then the adversary concludes the barrier at $L=y+r+\delta$, see 
Figure~\ref{fig:fg4ab} (b).
Clearly, the trajectory produced by $\cal{A}$ has length at least 
$3y-r+\delta$ while the offline algorithm uses a trajectory of at most 
$2y+ 2\delta$. Thus the competitive ratio of the algorithm is at least 
$(3y-r+\delta)/(2y+2\delta) \geq 3/2 -(r+2\delta)/(2x+2\delta) \geq  3/2 -(r+2\delta)/(2ir+2\delta) \geq 3/2 - \epsilon$ for sufficiently large $i$. 

To prove the second part of the theorem
observe that the algorithm that solves any gap in coverage with a triple  
from any potential left turning point has competitive ratio at most 
$3/2$. 
\end{proof}
In the remainder of the section, we assume that $L$, the length of the 
barrier segment, is known to the online algorithm. We first prove a lower 
bound on the competitive ratio of any online algorithm for the problem. 
\begin{theorem}
\label{thm:lb}
Assume that the online robot knows the length $L$ of the barrier  $[0, L]$. 
For any $0 < \epsilon \ll 1 $, the competitive ratio of any online algorithm is at least $\frac{5}{4} - \epsilon$.
\end{theorem}
\begin{figure}[!htb]
\begin{center}
\includegraphics[width=12 cm]{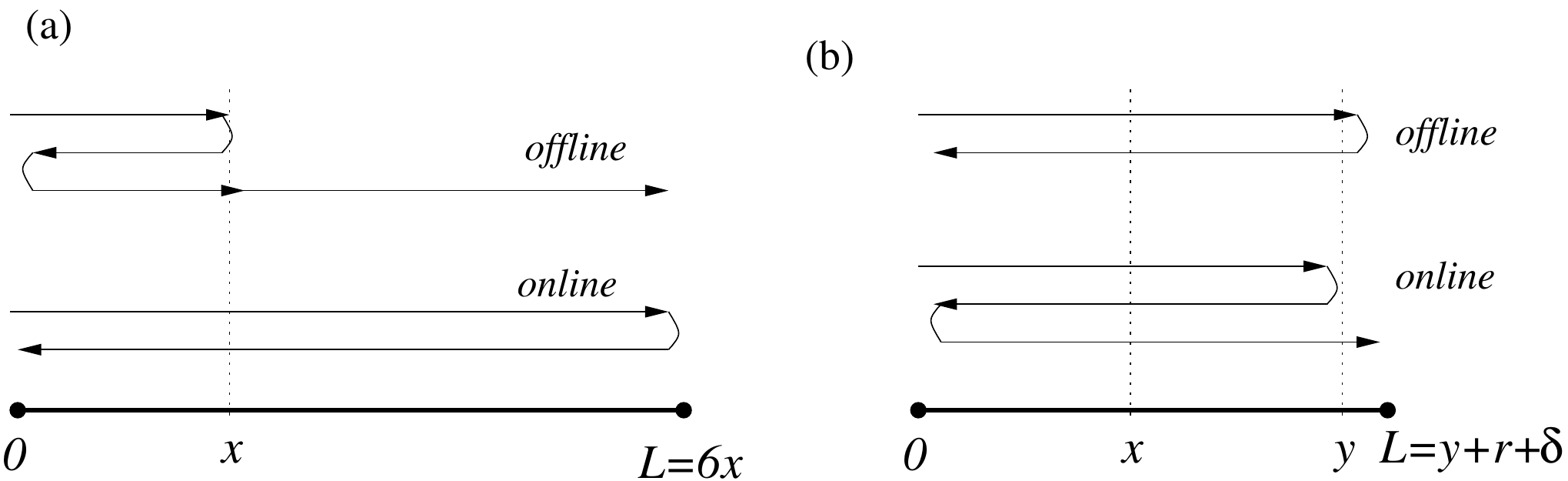}
\end{center}
\caption{Competitive ratio for the case when the robot does not know $L$.
In (a) the robot turns at $L$ or $y\geq 6x$, in 
(b)  the robot turns at $y< 6x$.}
\label{fig:fg4ab}
\end{figure}

\begin{figure}[!htb]
\begin{center}
\includegraphics[width=10cm]{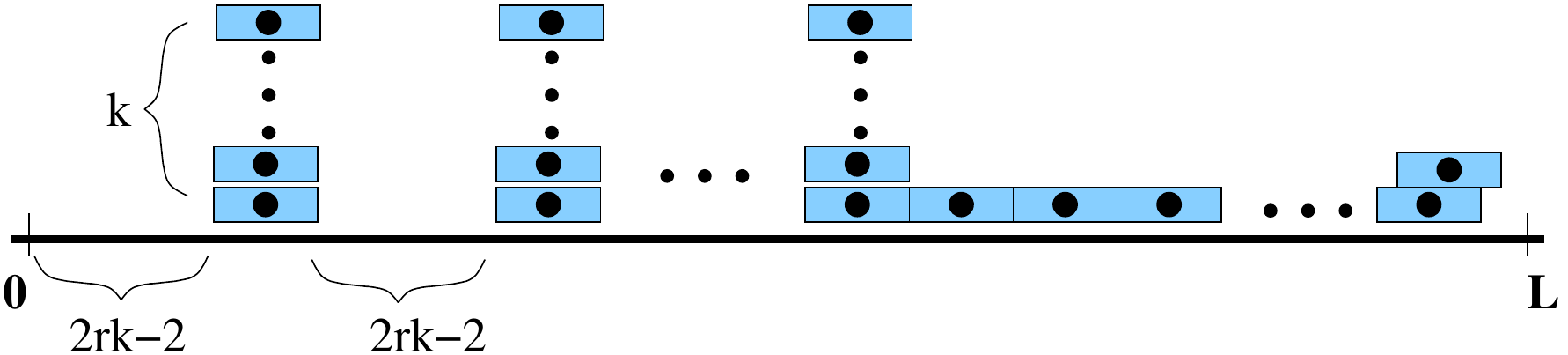}
\end{center}
\caption{Example for the lower bound on the competitive ratio}
\label{fig:lower-bound}
\end{figure}

\begin{proof}
Once again, we provide an adversary argument. We first give the high-level idea of the argument: the adversary creates a gap followed by a pile of $k$ stacked sensors that is just sufficient to cover the gap (see Figure~\ref{fig:lower-bound}). If the online robot decides to use a triple to cover the gap, the adversary repeats the pattern. If at some point, the online robot does not return with a triple to cover the gap, the adversary creates no more gaps. We then show that regardless of when the online robot decides not to cover a given gap with a triple, its trajectory is at least $\frac{5}{4}$ times the optimal trajectory. 

We proceed to give the details of the argument, constructing the input progressively as a function of the behaviour of the robot. Assume that the length of the barrier is $L$; the values of $k$ and $r$ will be specified later as functions of $L$. We repeat the following steps: We create a gap of size $2r(k-1)$ followed by $k$ sensors all at position $2r(k-1)+r$. It is clear that the first $k-1$ sensors can cover the first gap. We now observe the online robot. 
If the robot turns left at any point $y \geq x$ and does a triple to cover the gap before $x$, then as soon as possible after it returns to $y$, we repeat the same configuration: a gap of the same length followed by a stack of $k$ sensors. Until the robot 
turns left, we simply place sensors $2r$ apart. These steps are repeated
until we reach  $L-2r$, at which point we shorten the last gap if needed, and add as many sensors as needed to cover the preceding gap. We also add an additional (and last) sensor, to  force every algorithm, online and offline,  to come to position $L-r$ to cover the gap at $L$. 

Consider the online robot after it has finished its operation. Define the last gap that the robot covered with a triple to be 
the $i^{th}$ gap and let $x$ be the position of the stack of sensors immediately after this gap. Then $x \geq  2ikr-r$. Observe that the input instance has exactly one more gap after $x$, and the robot needs to cover the $i+1^{st}$ gap with a double. We denote the length of the online robot's trajectory by $T_A$. The length of the robot's trajectory until $x$ is at least $3x - 4ir + 2r$ and the length of the double to return to cover the $i+1^{st}$ gap is $2r +2(L-r-(x+2r))$. Therefore $T_A = x + 2L -4ir -2r$.
We compare the performance of the robot with two specific offline algorithms. The first algorithm does no triples and achieves coverage simply by doing one double. We denote the length of its trajectory by $T_1$; clearly,
$T_1 = 2L -3r$.
 The second algorithm uses the same trajectory as the online robot for all its triples, proceeds to cover the $i+1^{st}$ gap also with a triple followed by moving right to $L-r$. (Recall that after the $i+1^{st}$ stack of sensors, there are no gaps except the last gap at $L$.) We denote the length of the second algorithm's trajectory by $T_2$. Then the length of the trajectory for the second algorithm until $x+2rk$ is  $3(x+2rk)-4(i+1)r + 2r$ and the length of the final (single) segment  is $L-r -(x+2rk)$.
 \begin{equation*}
 T_2 = 2x+L+4rk-4ir -3r
 \end{equation*}
 
 Suppose $x \geq L/2 - 2rk$. Then:
$$ \frac{T_A}{T_1}  \geq  \frac{x+2L-4ir -2r}{2L} 
 >  \frac{L/2 - 2rk + 2L - 4ir -2r}{2L} 
 =  \frac{5}{4} - \frac{2rk + 4ir + 2r}{2L}$$

Suppose instead $x <  L/2 - 2rk$. Then
$$\frac{T_A}{T_2}  \geq  \frac{x+2L-4ir -2r}{2x+L+4rk-4ir - 3r}
 > \frac{x+2L-4ir -2r}{2x+L+4rk}
 = \frac{x + L/2 +2rk}{2x+ L + 4rk} +$$ 
$$+\frac{3L/2 - 2r(k+1) -4ir}{2x+L+4rk} 
 >  \frac{1}{2} + \frac{3L/2 - 2r(k+1) -4ir}{2L}
 =  \frac{5}{4} - \frac{2rk+4ir+2r}{2L}$$
By choosing $r = k = L^{1/3}$, we obtain $i = O(L^{1/3})$. This gives a lower bound on the competitive ratio of the online robot of $\frac{5}{4} - O(L^{-1/3})$. By choosing $L$ to be as large as needed, the theorem follows. 
\end{proof}

The optimal offline algorithm suggests that if the online robot stops doing 
triples too soon, it may be beaten by an algorithm that does perhaps 
just one more triple  which avoids the double  at the end. 
However if it keeps doing triples for too long, it may be beaten by an 
algorithm that does fewer triples. It is natural to ask whether there 
an optimal fraction of the segment such that the online robot can decide in 
advance to do triples only until then. We say that an online algorithm 
has a {\em fixed switching point} $z$ if it covers each gap before $z$ 
with a triple, and all gaps after $z$ with the final double. Therefore, 
the online robot turns left at most once after $z$, and if it does, it turns
at position $L-r$ to do the final double. We show below a tight bound on the 
competitive ratio of an online algorithm with a fixed switching point.
\begin{theorem}
\label{thm:fixed-switching-point}
Assume that the robot knows the length $L$ of the barrier 
$[0, L]$.
\begin{enumerate}
\item 
For any $0 < \epsilon \ll 1$, the competitive ratio of any online algorithm with a fixed switching point is at least $\frac{4}{3} - \epsilon$.
\item 
There is an online algorithm with fixed switching point with competitive ratio at most $\frac{4}{3} $.
\end{enumerate}
\end{theorem}
\begin{proof}
Consider an online robot that has a fixed switching point $z$ such that $0 \leq z \leq L$.  Suppose $z \leq 2L/3$, and let $k = \lfloor (z+r)/2r \rfloor$. Consider an input consisting of $k$ attached sensors that clearly cover the barrier until $2kr$,  then the next sensor at position $2kr+3r/2$. Clearly, there is a gap of length $r/2$ between the $k^{th}$ and $(k+1)^{st}$ sensors. Next place a sensor at position $2kr+3r$ and sensors in attached position after this until almost the endpoint with the usual gap at the end of the barrier. Thus the only two gaps are between the $k^{th}$ and $(k+1)^{st}$ sensors, and the one at the end. However,  since $z < 2kr+r$, the online robot cannot turn back at the $(k+1)^{st}$ sensor to do a triple to cover the gap before that sensor, and must instead cover both gaps with a double. Therefore, the online robot's trajectory is of length  at least $2L-2kr-3r \geq 2L-z-4r$, while the optimal trajectory contains a triple to complete the gap after $z$, and has length at most $L$. Since  $z \leq 2L/3$, the competitive ratio of the algorithm is at least $4/3 - 4r/L > 4/3 - \epsilon$ for $4r/L < \epsilon$. 

Suppose instead that $z > 2L/3$ and let $k = \lceil z/2r \rceil$. Now consider an input that has a gap from the beginning and then a stack of $k$ sensors at $z$ that suffice to cover the gap before $z$, followed by the next sensor at $2kr + 3r/2$ and attached sensors after that until almost the endpoint of the barrier as before. Observe that there are three gaps: the gap before $z$ that can be solved by turning left at $z$ and doing a triple,  the gap between the $k^{th}$ and $(k+1)^{st}$ sensors after $z$, and the gap at the end (If $z$ is so close to the right endpoint that there is no room for two gaps after $z$, then it is easy to show that the competitive ratio is actually $\approx 3/2$.) Thus the online robot must turn at $z$ to cover the gap with a triple, and cover the remaining gaps with a double, with a trajectory of length at least $2L+z-8r$, while the optimal trajectory has length $2L-3r$. Since $z > 2L/3$, the competitive ratio of the algorithm is at least $4/3 - 4r/L$. By choosing $4r/L<\epsilon$ for any given $\epsilon$, we obtain a lower bound of $4/3 - \epsilon$ for the competitive ratio.

Consider now an online algorithm $A$ with a fixed switching point $2L/3$: it turns left at every potential left turning point before $2L/3$ to cover the preceding gap with a triple. After it passes $z$ however, it will only turn left at the end and do a double if necessary. We claim that $A$ has competitive ratio $4/3$. Suppose the optimal offline algorithm $O$ does at least one fewer triple than $A$. Then it is not difficult to see that 
\begin{equation*}
\frac{T_A}{T_0} \leq \frac{8L/3 - 8r}{2L-3r} < \frac{4}{3}.
\end{equation*}
 If instead the optimal algorithm does at least one more triple than $A$, then it is not difficult to see that
 \begin{equation*}
\frac{T_A}{T_0} \leq \frac{4L/3 - r}{L} < \frac{4}{3}. 
\end{equation*}

Thus in both cases the competitive ratio of the algorithm is at most $4/3$, 
matching the lower bound.
\end{proof}

Thus, by deciding in advance a switching point at which to stop doing triples, it is impossible to derive an online algorithm that matches the lower bound of Theorem~\ref{thm:lb}. We now specify {\bf AdaptiveOnline}, an online algorithm 
for  a robot which, when starting at location $0$, relocates sensors on 
the segment $[0,L]$ to achieve complete barrier coverage. We  calculate  an upper bound on the competitive ratio of this 
online algorithm 
and prove that it asymptotically matches the lower bound of $5/4$ from Theorem~\ref{thm:lb}.
Clearly, an online algorithm can calculate the coverage balance  of any 
sensor it encounters. We now describe two functions for the online robot 
used in the algorithm. 
The first function is  called  {\em walk-in-surplus} and is defined as 
follows: When at a potential left-turning point (or the start of the barrier) 
the robot moves right picking  up sensors having a positive 
coverage balance and 
deposits them $2r$ apart (as the optimal offline algorithm constructing a 
fully-stretched solution would do), until reaching a point $x$ such that the 
last sensor it dropped was at location $x-2r$, and no sensors were encountered 
in the interval $[x-2r, x]$.  Observe that at such a position $x$, 
the robot knows that $x$ is a potential right-turning point. 
The function then returns the value $x$. 
The second function is  called {\em walk-in-deficit}: When 
first time at a potential right-turning point, 
robot moves right  picking up sensors 
with negative coverage balance on its way until it reaches a sensor 
with negative coverage balance greater than $-2r$, or  
balance exactly $-2r$ and collocated with the next sensor. Thus, this 
is a potential left turning point $y$; the function then returns the value $y$. 
The functions a {\em triple}, and  
a {\em double} 
behave the same way as in the offline algorithm.
The main challenge for the online algorithm is to determine, when 
it reaches a potential left turning point, whether to do a triple at 
this point, or to switch to solving the remaining segment as part of the 
final double.  

We specify our adaptive online algorithm as a recursive procedure 
{\bf AdaptiveOnline(t,L,r)} in which $[t,L]$ is the  
subinterval on which the robot has not yet travelled, and barrier coverage 
remains to be achieved, and $r$ is the range of sensors. 
To calculate the coverage of $[0,L]$ we execute 
 {\bf AdaptiveOnline(0,L,r)}. We assume that there is a gap at position $0$; if not, we simply execute the walk-in-surplus function until reaching a potential right turning-point $x$ and then call {\bf AdaptiveOnline} on the segment $[x-r, L]$. It is clear that the initial part of the trajectory executed until $x$ is optimal, and cannot increase the competitive ratio on the entire input. We give the pseudocode for the algorithm below.
\begin{tabbing}
xx\=xx\=xxx\=xx\=xx\=xx\=xx\=xx\=xx\=xx\=xx\= \kill
Algorithm  {\bf AdaptiveOnline} $(t,L,r)$;\\
\line(1,0){440}\\
\> {\bf Input:}
$t$, $L$, the subinterval being solved,  with a gap at $t$ and
 $r$ is the range of sensors\\
\> {\bf Output:} \>\>\>\>the moves of the robot;\\
\> {\bf Variables:}: \\
\>\>$x$; // the current position\\
\>\>$T$; // current trajectory length\\
\>\> $\gamma_i$; // ratio trajectory/distance at left-turning point in iteration $i$\\
\>\> $\beta_i$; // ratio trajectory/distance at right-turning point in iteration $i$\\
\>{\bf functions:} walk-in-surplus, walk in deficit, triple and double\\
\line(1,0){440}\\
 \> $x \la t+ r$;  $T \la 0$;  $i \la 0$; // initialization of variables\\
\>{\bf repeat} \\
\> \> $i \la i+ 1$; // iteration of loop\\
\>\>$b_i \leftarrow x$; // potential right-turning point \\
\>\>$\beta_i \la  (T+r)/b_i$ $\;$// ratio at start of possible triple\\
\>\>$a_i \la $ walk-in-deficit; // potential left-turning point \\
\>\>$T \la T+r + 3(a_i-b_i)$; // trajectory if triple is done \\
\>\>$\gamma_i \la T/a_i$; // ratio at end of possible triple\\
\>\>{\bf if}  $\gamma_i a_i -a_i >  L - t$ {\bf break} $\;$// exit the loop\\   
\>\> {\bf else} \\
\>\>\>do a triple  on segment $[b_i,a_i]$,\\
\>\>\> $x \la $walk-in-surplus; // gap starting at $x-r$ \\
\>\>\> $T \la T + (x-r-a_i)$; //update trajectory until start of gap\\
\> \> \> {\bf If} $x < L$ and $T/(x-r) \leq 2.5$ {\bf then}  AdaptiveOnLine(x-r,L,r);\\
\>{\bf until} $x=L$;\\
\> {\bf if} ($L$ not reached) {\bf then}\\ 
\>\>\> {\bf do} a double  (to $L-r$ and back to $b_i$); \\
\>\>\> $T \la T + (L-a_i)-(a_i-b_i)$;
\end{tabbing}

 The key idea is as follows: First, the online robot keeps track of  the ratio between its trajectory so far versus the distance it has covered. If it discovers that this ratio is less than $5/2$, then it "forgets about" the segment covered so far (it will be shown that it has achieved a competitive ratio of at most $5/4$ for this part), and restarts its computations. 
The ratio between its trajectory and distance travelled so far is computed only at potential left and right turning points. 
 Secondly, when it reaches a potential left-turning point, the online robot 
calculates the cost of the triple: the difference between its trajectory if it executes the triple and the distance covered so far. If this difference is too high, it decides to stop doing triples, and finish by doing a double.

Observe that before making a recursive call, at least one gap is covered by the robot. Since the number of gaps is finite, 
the algorithm terminates. It is also clear that {\bf AdaptiveOnLine} 
constructs 
a trajectory that results in barrier coverage. It remains only to analyze the competitive ratio of the trajectory length.  Let $T_A(I)$ and $T_o(I)$ be the lengths of the trajectories of the algorithm
{\bf AdaptiveOnline} and the optimal offline algorithm on an input instance $I$ respectively. We  prove a bound of $5/4$ on $max_I \{ T_A(I)/T_o(I) \}$, 
thereby matching the lower bound of Theorem~\ref{thm:lb}. 

Fix an input instance $I$. Observe that the algorithm  {\bf AdaptiveOnline} 
partitions the segment $[0,L]$ into sub-segments that are solved in each 
recursive call of the algorithm. We call each of these sub-segments an 
{\em epoch}; let $n$ be the number of epochs, such that while traversing 
epoch $j$, there is no recursive call. 
Let $T_j$ be the length of the  the trajectory of the online robot in epoch 
$j$, and let $O_j$ be the length of the trajectory of the optimal offline 
robot in the same epoch. Every epoch starts with a gap, and in every epoch 
except possibly the last, the mobile robot does triples from the first 
encountered left-turning point to cover gaps. 
\begin{lemma} \label{lemma:early-epochs}
$T_j/O_j \leq 5/4$ for $1 \leq j \leq n-1$
\end{lemma}

\begin{proof}
Let the length of epoch $j$ be $\ell_j$. Since the epoch starts with a gap, the optimal algorithm either does the same thing as the mobile robot in the epoch, in which case $T_j/O_j = 1$ or it does part of the epoch using a double in which case $T_j/O_j \leq  \frac{T_j}{2 \ell_j}$. Since in each of the epochs $j$ for $1 \leq j \leq n-1$, the recursive call is made only when $T_j/\ell_j  \leq 2.5$, we have $T_j \leq 2.5 \ell_j $. The lemma follows. 
\end{proof}

\begin{lemma} \label{lemma:same-difference}
The difference between the trajectory length and the distance covered on the segment stays the same  while executing {\em walk-in-surplus} or {\em walk-in-deficit}. 
\end{lemma}
\begin{proof}
Let $T$ be the trajectory length at point $z$ just before executing walk-in-surplus (or walk-in-deficit) and let $T'$ be the trajectory at point $z'$ at the end of its execution. Then $T' - z' = (T + z'-z) - z' = T-z$.
\end{proof}

\begin{lemma} \label{lemma:last-epoch}
$T_n/O_n \leq 5/4$
\end{lemma}
\begin{proof}
In epoch $n$, there are no recursive calls. For simplifying the analysis, we assume $t=0$ and $L=1$. Assume the loop was exited in iteration $i$. Observe that the loop can be exited either because $x=1$ or because $T-a_i> 1$. 

We first consider the case when the loop is exited because $x=1$. In this case, we know that $\gamma_i a_i - a_i \leq 1$ and the online robot walked in surplus after this till the end. Thus, the optimal offline algorithm could not have done more triples than the online robot. Consider an optimal offline algorithm that does fewer triples. Then 
\begin{equation}
\frac{T_n}{O_n} \leq \frac{T_n}{2} = \frac{ \gamma_i a_i + (1 - a_i)}{2} \leq \frac{2}{2}= 1
\end{equation}
Thus $T_n/O_n = 1$ in this case. 

Next we consider the case when the loop is exited in iteration $i$ because $T-a_i  > 1$. If $i=1$, since $T= r + 3(a_i-r)$, this implies $2a_i - 2r > 1$. Observe that the online robot does no triples at all, and does only a double and has a trajectory of length $2-r$. An algorithm that does the triple in the segment $[r, a_i]$ has length at least $r + 3(a_i-r) + 1-a_i = 1+ 2a_i -2r > 2$ which means it is sub-optimal. Therefore, in this case again $T_n/O_n = 1$. 

Suppose instead the loop is exited in iteration $i >1$ because $T-a_i = \gamma_i a_i - a_i > 1$.   Denote by $d$, the distance between $a$ and $b$, that is, $d=a-b$. Then 
\begin{equation}
\beta_i = \frac{ \gamma_i a_i - 3d}{a_i -d}
 \label{eqn:beta}\end{equation}
Next observe that since at the end of iteration $i-1$, the value of $T/(x-r) > 2.5$,  we obtain $3 \geq \beta_i = (T+ r)/x > 2.5$. It follows now from Equation~\ref{eqn:beta} that 
\begin{equation}
\gamma_i > 2.5
\label{eqn:alpha}
\end{equation}
Since the loop was not exited in iteration $i-1$, it must be that $\gamma_{i-1}a_{i-1} - a_{i-1} \leq 1$ and since the only moves between $a_{i-1}$ and $b_i$ are achieved by walk-in-surplus and walk-in-deficit, it follows from Lemma~\ref{lemma:same-difference} that $\beta_i b_i - b_i  = \gamma_{i-1}a_{i-1} - a_{i-1} \leq 1$. Substituting $b_i = a_i-d$, and using Equation~\ref{eqn:beta}, we obtain:
\begin{equation} 
\gamma_i a_i - a_i - 2d \leq 1.
\label{eqn:prev}
\end{equation}
Finally, recall that  
\begin{equation} 
\gamma_i a_i - a_i > 1
\label{eqn:curr}\end{equation}
For reasons that will be clear later, we now prove the following claim:
\begin{claim} \label{extra}
$\gamma_i a_i - 2a_i - d \leq 0.5$
\end{claim}
\begin{proof}
Suppose $a_i \leq 0.5 + d$. Then since $\gamma_i \leq 3$, we conclude that
$\gamma_i a_i - 2 a_i - d \leq  3a_i - 2a_i - d = a_i -d \leq 0.5$. If instead that $a_i > 0.5 + d$, then it follows from Equation~\ref{eqn:prev} that
$\gamma_i a_i - 2a_i - d < 1 - (a_i - d) < 0.5$ 
\end{proof}

Now consider the value of $T_n$. We have 
\begin{equation}
T_n = \gamma_i a_i + 2 (1-a_i) - d = 2 + \gamma_i a_i - 2a_i - d
\end{equation} 

Suppose the optimal algorithm did fewer triples than {\bf AdaptiveOnline}. Then 
\begin{equation}
\frac{T_n}{O_n} \leq \frac{T_n}{2} = \frac{2 + \gamma_i a_i -2a_i - d}{2} \leq \frac{5}{4}
\end{equation}
where the last inequality follows from Claim~\ref{extra}.
If instead the optimal algorithm did more triples that {\bf AdaptiveOnline}, its trajectory is of length at least 
$\gamma_i a_i + 1 - a_i$. Therefore
\begin{equation}
\frac{T_n}{O_n} \leq \frac{T_n}{\gamma_i a_i + 1 - a_i} \leq \frac{T_n}{2} \leq \frac{5}{4}
\end{equation}
where the second inequality uses~(\ref{eqn:curr}).
\end{proof}
Lemmas, \ref{lemma:early-epochs}, \ref{lemma:last-epoch} 
show that in each epoch the competitive ratio is at most $5/4$. Thus we get 
the following theorem.
\begin{theorem} 
{\bf AdaptiveOnline} is an online algorithm for barrier coverage of a line segment of known length and has competitive ratio at most  $5/4$, and is therefore optimal. 
\end{theorem}    
\bibliography{refs}
\bibliographystyle{plain}
\end{document}